\documentclass[journal]{IEEEtran}
\usepackage{epsfig,rotating,setspace,latexsym,amsmath,epsf,amssymb,amsfonts,bm,theorem,epstopdf,cite,mathtools,caption,subcaption,enumerate,longtable,accents,bbm,algorithm,algorithmic,graphicx,epsf,authblk,url,color,multirow,comment,tabularx,dsfont,physics,tikz,dsfont}
\usepackage[inline]{enumitem}

\newtheorem{theorem}{Theorem}
\newtheorem{lemma}{Lemma}
\newtheorem{corollary}{Corollary}

\newtheorem{remark}{Remark}

\newenvironment{Proof}[1]{\medskip\par\noindent{\bf Proof:\,}\,#1}{{\mbox{\,$\blacksquare$}\par}}
\newcommand{\figref}[1]{\figurename~\ref{#1}}
\IEEEoverridecommandlockouts
\allowdisplaybreaks

\title{All-out Attack: Optimal Block Withholding Under Pay-Per-Share Scheme}

\author{Mustafa Doger \qquad Sennur Ulukus\\
\normalsize Department of Electrical and Computer Engineering\\
\normalsize University of Maryland, College Park, MD 20742\\
\normalsize  \emph{doger@umd.edu} \qquad \emph{ulukus@umd.edu}}

\date{}

\begin{document}

\maketitle
\begin{abstract}
Classical Block Withholding (BWH) attacks have been extensively studied in block-dependent reward schemes, where pool members are compensated upon a block discovery within the pool. However, most contemporary mining pools operate under share-based schemes, wherein participants are paid immediately upon submission of valid shares. In this paper, we analyze BWH under Pay-Per-Share (PPS) for Nakamoto-style blockchains and prove that these mechanisms are not incentive compatible, contrary to conclusions drawn under fixed difficulty in prior literature. Under PPS, the optimal strategy for a BWH attacker is the All-out Attack (AoA): the adversary allocates its entire hashpower toward the victim pool, submitting only partial Proof-of-Work shares (pPoW) while withholding all valid blocks, i.e., full Proof-of-Work (fPoW). 

Prior to the first difficulty adjustment, the adversary incurs negligible loss from withheld fPoWs. After the adjustment reduces block difficulty, the adversary either generates more pPoWs per unit time when the pPoW difficulty is reduced accordingly or, if the pPoW difficulty is held fixed, earns a higher reward per share. In both cases, it achieves a post-adjustment reward rate of $\frac{\alpha}{1-\alpha}$ per target epoch, compared with the honest baseline rate of $\alpha$. Remarkably, this gain matches the theoretical upper bound achieved by optimal selfish mining in Nakamoto consensus under perfect network influence. The results further indicate that BWH is substantially more profitable under PPS than under mainstream block-dependent payout schemes, even when compared with advanced BWH variants. Honest miners benefit at the same rate as the adversary per unit hashpower, while the victim pool operator bears all losses, paying out-of-pocket for pPoW submissions without receiving fPoW compensation in return. Finally, we argue that advanced BWH variants such as Fork After Withholding (FAW) yield no additional profit under PPS.
\end{abstract}

\section{Introduction}
Nakamoto consensus enables distributed participants in a network to reach agreement in a permissionless and trustless manner. By solving cryptographic puzzles, a process known as \textit{mining}, participants, referred to as \textit{miners}, become eligible to append a new block containing the embedded solution, termed \textit{Proof-of-Work (PoW)}, to the chronological chain, collectively called a \textit{blockchain}. Nakamoto incentivizes mining by rewarding miners for successfully extending the blockchain and broadcasting newly mined blocks to other participants~\cite{btc-whitepaper}.

As the time required for a miner to discover a new block is inversely proportional to its computational power and exhibits high variance, miners often collaborate by forming mining pools. Pools provide participants with a more stable revenue stream by distributing block rewards among members proportional to their contributions. Historically, pools utilized block-dependent reward mechanisms such as Proportional (PROP) or Pay-Per-Last-N-Share (PPLNS), under which block rewards are distributed among members upon the successful mining of a block. In these schemes, pool participants collectively bear the variance associated with block discovery times. In contrast, most modern pools adopt share-based reward mechanisms such as Pay-Per-Share (PPS) and Full-PPS (FPPS), under which members receive immediate compensation for each share submitted. Consequently, the risk of mining time variance is entirely borne by the pool operator~\cite{rosenfeld2011analysisbitcoinpooledmining}.

Eyal and Sirer~\cite{selfish-mining} demonstrated that Nakamoto consensus is vulnerable to incentive-based attacks: a miner can temporarily withhold the blocks it mines according to a selfish mining strategy and subsequently releases them to replace those mined by other participants. Such an attack can increase the fraction of the blocks mined by the attacker and after the blockchain adjusts the difficulty, the attacker begins to accrue profits~\cite{grunspan2019-profitability-selfish-mining}. Similarly, a member can attack its pool by withholding the valid pool blocks (full PoW, fPoW) it mines, known as \textit{Block Withholding} (BWH) attack~\cite{rosenfeld2011analysisbitcoinpooledmining}. In  2014, Eligius mining pool announced that it was subject to a BWH and lost 300 BTC~\cite{EligiusBlockWithholding2014}. 

Courtois and Bahack~\cite{courtois2014subversiveminerstrategiesblock} showed that an adversarial miner allocating a fraction of its hashpower to solo honest mining while directing the remainder toward a pool, where it withholds fPoW blocks and submits only partial PoW shares (pPoW), can achieve additional profits. Later studies~\cite{power_splitting_pools,optimal-bwh} re-analyzed the attack under more general conditions using rigorous formalisms. There are three natural extensions of BWH: \textit{Fork After Withholding} (FAW)~\cite{fork_after_witholding_attack}, \textit{Power Adjusting Withholding} (PAW)~\cite{power_adjusting}, and \textit{Temporary PAW} (T-PAW)~\cite{doger2026temporarypoweradjustingwithholding}. Each extension offers progressively greater flexibility to the attacker, enabling increased adversarial profits. Similarly, extensions of BWH, FAW, and PAW, including those in~\cite{block_withholding_delay,return-after-withholding,uncle_block_attack,Wang2023EFAW,if_you_cant_beat_pay_them,bm-paw_attack,Zhu2022_revisit_faw}, explore the attack from multiple perspectives. Another study investigates strategic benefits of pPoW withholding~\cite{chang2020sharewithholdingattackblockchain}.

\subsection{Related Works and Our Contributions}
All works mentioned above on BWH~\cite{courtois2014subversiveminerstrategiesblock,power_splitting_pools,optimal-bwh,fork_after_witholding_attack,power_adjusting,doger2026temporarypoweradjustingwithholding,block_withholding_delay,return-after-withholding,uncle_block_attack,Wang2023EFAW,if_you_cant_beat_pay_them,bm-paw_attack,Zhu2022_revisit_faw} consider block-dependent payout schemes. However, most contemporary mining pools employ share-based schemes. The analyses presented for block-dependent payout schemes rely on the mainstream zero-sum game framework (revenue ratio maximization) and do not extend to share-based schemes. In share-based schemes, pool operator reserves become part of the strategic interaction and invalidate the zero-sum assumption.

In his analysis of the sabotage attack~\cite{rosenfeld2011analysisbitcoinpooledmining}, Rosenfeld claims that BWH against PPS does not increase the attacker’s rewards but instead aims to harm the pool operator. Similarly, works such as~\cite{schrijvers2016incentive,Saide_Zhu_2018_Mathematical_Foundations_of_Computing} also consider PPS schemes and assert that delaying the reporting of a fPoW yields no profit, concluding that the scheme is incentive compatible. These conclusions are drawn under a fixed block difficulty. By modeling difficulty adjustment, we prove otherwise in this paper. In fact, we demonstrate that the adversary should mount an All-out Attack (AoA): fully allocate all its hashpower toward the victim pool while never submitting any fPoW and regularly submitting pPoWs.

AoA exploits the block difficulty adjustment (DA) mechanism and benefits every miner in the blockchain network, while losses are borne solely by the victim pool operator(s). In essence, by withholding fPoW, the adversary forgoes only a negligible income from fPoW submissions prior to the first DA. After the first DA, the blockchain network reduces the fPoW mining difficulty because fewer blocks are mined per time unit. If the pool maintains a constant average ratio of pPoW to fPoW mined, implying that pPoW rewards remain fixed, then the pPoW difficulty also decreases. As a result, the adversary can generate more pPoWs per unit time and thereby accrue profits. Alternatively, if the pool does not adjust its pPoW difficulty in response to the reduced fPoW difficulty, the adversary generates the same number of pPoWs as before but receives higher rewards per pPoW, since pPoW rewards are determined relative to the difficulty of fPoW.

We prove that all other miners in the network benefit from the attack at the same rate as (or even slightly more than) the adversary mounting an AoA. In contrast, the victim pool operator(s) bear all losses: they pay the adversary out-of-pocket for pPoW submissions but never receive a fPoW as compensation. Although FAW/PAW/T-PAW are shown to be effective in block-dependent payout schemes, they do not benefit the attacker under PPS because releasing a block may trigger an increase in block difficulty, which reduces mining efficiency. However, if orphan blocks do not contribute to block difficulty, FAW can marginally mitigate losses for the pool operator, even though it does not directly benefit the adversary.

A related work~\cite{smart_contract_attack_pps} considers a solo mining adversary who bribes pool miners in Bitcoin to withhold the fPoWs they mine, analyzed within a zero-sum game framework. This setting is considerably more complex and arguably less realistic: it requires a smart contract on a third-party platform (e.g., Ethereum), the participation of pool members willing to accept bribes, and coordination between them and the adversary. In contrast, our attack requires no bribery or coordination: a single attacker mounts it alone while mining for the pool. Moreover, we characterize the optimal hashpower allocation under both fixed and variable share difficulty, and derive the revenue change of every entity, including the pool operator.

\section{System Model}
We consider a Nakamoto-style blockchain network operating under the longest-chain PoW protocol. The network consists of $n$ miners, each controlling an infinitesimal fraction of the total hashpower. We assume that $\alpha$ fraction of the total hashpower is controlled by colluding adversarial miners. For the remainder of this paper, all adversarial hashpower is assumed to be centralized under a single entity: the adversary. The remaining miners are referred to as honest miners, who follow the Nakamoto protocol faithfully; they mine on the tip of the longest chain and immediately disseminate any newly mined block to their peers. We refer to blocks that persistently form part of the longest chain in the long-run as \textit{canonical blocks} and define an \textit{epoch} as the time duration required for the longest chain to accumulate $D_0$ canonical blocks. For simplicity, we assume that the total block reward distributed across the network during a given epoch is normalized to 1.  

We assume that the blockchain operates under a Difficulty Adjustment Algorithm (DAA) similar to Bitcoin’s which ensures that each epoch lasts on average $\tau_0$. To achieve this, after each epoch, the DAA measures the actual elapsed time $\tau$, and rescales the block difficulty for the next epoch by a factor of $\frac{\tau_0}{\tau}$.\footnote{For the sake of convenience, we assume that, for a fixed hashpower, a higher difficulty results in fewer blocks being mined.} The change can be quantified using the \textit{block redundancy ratio}, which is defined as the total number of blocks mined during an attack renewal period divided by the total number of canonical blocks mined within that period. 

We consider a victim pool (excluding the adversary) that controls $\beta$ fraction of the total hashpower in the system and is targeted by the adversary satisfying the constraint $\alpha + \beta \leq 0.5$, the majority threshold for the pool's total hashpower with or without the adversary. The victim pool is managed by an operator who distributes PoW tasks to its members and receives pPoW/fPoW solutions. We consider two distinct pool reward payout schemes (block-dependent and share-based) which will be discussed in detail in the following sections.

\subsection{Pool Reward Payout Schemes}\label{sec::pay_scheme}
In this paper, we focus on PPS, a share-based scheme in which each valid share receives an immediate payment from the pool operator proportional to its expected statistical value of the coinbase reward. In contrast, under block-dependent schemes (e.g., PROP, PPLNS) miners are paid only upon the discovery of a block, based on their proportional contribution.

More specifically, consider that the current block difficulty requires $Z_f$ leading zeros for a block hash to be valid, whereas a valid share only needs $Z_p$ leading zeros, with $Z = Z_f - Z_p$ and $2^{-Z} \ll 1$. The corresponding mining difficulties are proportional to $2^{Z_f}$ and $2^{Z_p}$, respectively. Therefore, on average, it takes $2^Z$ shares before the pool discovers one block and receives one block reward, denoted as $R_b$. Consequently, each submitted share, regardless of whether it is a pPoW or fPoW solution, earns an immediate payment equal to 
\begin{align}
    R_s=R_b\times\frac{\textit{share difficulty}}{\textit{block difficulty}}=R_b2^{-Z}.
\end{align}

Note that PPS has a variant, FPPS, where each share receives an immediate average proportional share of transaction fees, in addition to its statistical coinbase reward share. Although our theoretical analysis of PPS in Section~\ref{sec::pps_analysis} considers only coinbase rewards, the results generalize to FPPS by treating the effective block reward as the sum of two components: the coinbase reward and the transaction fee component.

Since each share receives an immediate fixed payment, miners enjoy a stable cash flow, while all risk associated with the variance of block mining times is borne by the pool operator. Consequently, during unlucky periods when the pool fails to discover a block, the operator incurs losses; conversely, during lucky periods, the operator realizes additional profits. As such, PPS operators require substantial reserves and strong capital backing. Most modern pool operators implement this payout scheme in practice. Indeed, PPS/FPPS remains widely deployed in Bitcoin mining. For example, Foundry USA (approximately 25--28\% of the network hashrate), AntPool (approximately 19--21\%), F2Pool (approximately 12--14\%), and ViaBTC (approximately 10--11\%) collectively account for a substantial fraction of Bitcoin's total hashrate~\cite{mempool_mining}. Among these, Foundry USA employs FPPS~\cite{foundry_pool}, AntPool supports FPPS, PPS and PPLNS~\cite{antpool}, ViaBTC offers PPS+ and PPLNS~\cite{viabtc}, while F2Pool supports FPPS/PPS variants~\cite{f2pool}.

\subsection{Incentive Analysis and Revenue Changes}
Classical BWH studies focus on block-dependent payout schemes in which coinbase rewards are distributed among pool members upon block discovery. In such settings, zero-sum game analysis with revenue ratio maximization is a standard approach for studying incentive compatibility, since increasing rewards for one mining entity directly reduces those of another as coinbase rewards constitute the sole asset under consideration. In contrast, in share-based schemes, pool operators maintain capital reserves to pay members for finding pPoWs even when the pool does not consistently discover fPoWs and receive coinbase rewards. Consequently, a zero-sum game framework does not apply; therefore, we adopt an alternative analytical technique, i.e., revenue change analysis, commonly employed in the literature to study incentive attacks such as selfish mining \cite{grunspan2019-profitability-selfish-mining,intermittent_mining} and BWH \cite{doger2026temporarypoweradjustingwithholding}.

In short, in profitability analysis, we compare the average revenue of each mining entity under the adversarial attack with the average revenue they would obtain in the absence of the attack. To this end, we define a reference time $t=0$ at which a new epoch begins and the adversary initiates the attack; prior to that, all entities including the adversary mine honestly. We denote $\Delta_x(t)$ as the revenue change of mining entity $x$ at time $t$, defined as the difference between its revenue under the attack and without it. Similarly, we define $\overline{\Delta}_x(t) = \frac{\Delta_x(t)}{f_x}$ as the relative revenue change for entity $x$, where entity $x$ controls $f_x$ fraction of total hashpower. The relative revenue change metric is useful because it quantifies how much additional revenue an entity gains per unit of hashpower it controls. In the remainder of this paper, whenever we refer to revenues, we mean their expected values.

\section{Classical Block Withholding Attacks Under Block-Dependent Payout Schemes}
Let us briefly revisit classical BWH attack under a block-dependent payout scheme, in which pool miners are paid according to their share of contributions when the pool discovers a block. Consider a victim pool controlling $\beta$ fraction of the total system hashpower, and an adversary controlling $\alpha$ fraction, with $\alpha + \beta \leq 0.5$. Suppose the adversary joins the victim pool using $p$ fraction of its hashpower, submitting pPoWs it encounters but never submitting any fPoW. The remaining ($1 - p$) fraction of its hashpower is used to mine individually and honestly (referred to as “solo” mining).

Under such an attack with a block-dependent payout scheme, we can directly apply the zero-sum game principle. The revenue ratios for the adversary ($\rho^{p}_{A,BWH}$), honest victim pool members ($\rho^{p}_{P,BWH}$), and the rest of the honest miners ($\rho^{p}_{R,BWH}$) are given by
\begin{align}
\rho^{p}_{A,BWH}&=\frac{\alpha(1-p)+\beta r_1}{1-\alpha p},\label{eq::adv_rho}\\
\rho^{p}_{P,BWH}&=\frac{\beta (1-r_1)}{1-\alpha p},\\
\rho^{p}_{R,BWH}&=\frac{1-\alpha-\beta}{1-\alpha p}.
\end{align}
Here, $r_1$ denotes the average adversarial share of rewards per block discovered by a pool member and is given by
\begin{align}
    r_1=\frac{\alpha p}{\beta+\alpha p}\label{eq::r}
\end{align}
for most practical block-dependent schemes, including PROP, PPLNS, and Score-based. It is straightforward to show that the block redundancy ratio under BWH is
\begin{align}
 \delta^{BWH}_{p}&=\frac{1}{1-\alpha p}.
\end{align}

BWH has three natural extensions: FAW, PAW and T-PAW: In \textbf{FAW}, the adversary releases a withheld fPoW when any miner outside the victim pool mines a block. This triggers a fork race, which can increase adversarial rewards compared to classical BWH \cite{fork_after_witholding_attack}. \textbf{PAW} generalizes FAW by allowing the adversary to adjust $p$ to $p_2$ whenever it encounters and withholds a fPoW, thereby potentially increasing its profit. Notably, setting $p_2 = p$ reduces PAW to FAW \cite{power_adjusting}. \textbf{T-PAW} further generalizes PAW by imposing a maximum withholding duration of $T$ time units. If no block is mined during this period, the adversary releases the fPoW. Taking $T \to \infty$ reduces T-PAW to PAW \cite{doger2026temporarypoweradjustingwithholding}.

An adversary can choose $p$ (and also $p_2$ and $T$ in advanced variants) to maximize its revenue ratio, thereby increasing its profit relative to honest mining. The \textbf{revenue change} of each entity $x$ under this zero-sum game is given by
\begin{align}
   \Delta_x(t)=\begin{cases} 
  (\rho_x-f_x\delta)\frac{t}{\delta\tau_0}, & t\leq \delta\tau_0, \\ 
  \rho_x-f_x\delta+(\rho_x-f_x)\frac{t-\delta\tau_0}{\tau_0}, & t>\delta\tau_0. 
\end{cases}
\end{align}

\section{Optimal BWH Under PPS: All-out Attack}\label{sec::pps_analysis}
In this section, we analyze the profitability of BWH under the PPS scheme. We show that the optimal BWH attack in PPS differs substantially from the mainstream zero-sum game analysis (revenue ratio maximization) commonly applied to block-dependent schemes. Specifically, under PPS, the adversary allocates all its mining power exclusively to the victim pool and never releases a withheld fPoW, a strategy we term \textbf{All-out Attack} (AoA). We argue that advanced variants of BWH do not yield additional adversarial profit under PPS; however, FAW can marginally mitigate losses for the pool operator. Our results extend to the FPPS scheme by treating the effective block reward as the sum of coinbase and transaction fees.

With the AoA, victim pool operators under the PPS scheme suffer complete harm: the adversary allocates all its hashpower toward the victim pool but never releases any block. Consequently, the victim pool effectively pays the adversary out-of-pocket for every pPoW share. Since no block is published by the adversary, difficulty drops after the first epoch. Following the DA, the adversary begins to accrue profits. We now provide a rigorous analysis substantiating these claims. Let us assume that at time $t=0$, the adversary joins the victim pool with a fraction $p$ of its hashpower and initiates the classical BWH attack. Under the PPS scheme, we show that the optimal adversarial power allocation is $p^* = 1$, i.e., the All-out Attack.

We first assume that the mining pool maintains a fixed relative difficulty gap between fPoWs and pPoWs for the adversary, i.e., $Z = Z_f - Z_p$ remains constant and any change in fPoW difficulty is accompanied by an identical adjustment in pPoW difficulty. Under this assumption, the reward per share $R_s$ remains constant throughout the analysis. However, if the fPoW difficulty decreases, an honest pool miner allocating the same amount of hashpower will submit shares at a higher rate. We first analyze this setting for simplicity before extending the results to a more practical setting.

In practice, mining pools typically employ \emph{variable share difficulty} (vardiff) to maintain an approximately constant share submission rate, e.g., one share every few seconds. Vardiff is primarily an operational mechanism that reduces server overhead and improves hashrate estimation; it does not affect a miner's expected payout. If the network block difficulty decreases while the share difficulty remains fixed, each submitted share becomes more valuable because the PPS reward is proportional to the ratio between the share difficulty and the block difficulty. Conversely, if the share difficulty is adjusted proportionally with the fPoW difficulty, the reward per share remains unchanged, while miners submit more shares per time. The main results of this paper hold under both assumptions, with Theorems~\ref{thm:AoA} and~\ref{thm:AoA_vardiff} establishing the optimal attack strategy under each assumption. Unless stated otherwise, we neglect pool fees; see Remark~\ref{rem::fees} for the case with fees.

\begin{lemma}\label{lemma::first_DA}
Assume the adversary mines honestly prior to $t=0$, and the blockchain adjusts difficulty at $t=0$, initiating a new epoch (the first). Under the PPS scheme, if the adversary launches a classical BWH attack against a victim pool using fraction $p$ of its hashpower from $t=0$ onward, the duration of the first epoch is
\begin{align}
    t_1=\delta^{BWH}_{p}\tau_0=\frac{\tau_0}{1-\alpha p}.
\end{align}
At the end of this epoch, the adversary’s revenue change is
\begin{align}
    \Delta_A^{BWH,H}(t_1)=-\delta^{BWH}_{p}\alpha p2^{-Z}\approx 0,
\end{align}
whereas the revenue changes of the honest miners in the victim pool and those outside are zero, i.e.,
\begin{align}
    \Delta_{H_P}^{BWH,H}(t_1)=\Delta_{H_R}^{BWH,H}(t_1)=0,
\end{align}
and that of the victim pool operator is
\begin{align}
    \Delta_{Po}^{BWH,H}(t_1)=-\delta^{BWH}_{p} \alpha p (1-2^{-Z}).
\end{align}
\end{lemma}

\begin{Proof}
The duration of the first epoch is scaled by $\delta^{BWH}_{p}$ because a fraction $\alpha p$ of the total system hashpower is effectively wasted due to the adversary’s withholding behavior. Note that, in the absence of attack, an entity controlling $f_x$ fraction of the total network hashpower would earn $f_x$ rewards over time interval $\tau_0$. Thus, if the adversary were mining honestly during $[0, \delta^{BWH}_{p} \tau_0]$, its reward would be $\alpha \delta^{BWH}_{p}$, that of honest miners in the victim pool would be $\beta \delta^{BWH}_{p}$, and that of the remaining honest miners would be $(1 - \alpha - \beta) \delta^{BWH}_{p}$, all measured over time interval $\delta^{BWH}_{p} \tau_0$. Hence, if during $[0, \delta^{BWH}_{p} \tau_0]$ the adversary allocates fraction $p$ of its hashpower honestly toward the victim pool, it would earn $\delta^{BWH}_{p} \alpha p$ PPS rewards.

Under the attack, we partition the reward analysis into two components: \begin{enumerate*}
    \item \textbf{Coinbase rewards:} The portion attributable to coinbase rewards earned by the adversary (solo), the victim pool operator, and the remaining honest miners.
    \item \textbf{PPS rewards:} The portion corresponding to PPS rewards transferred from the victim pool operator to its members, including both honest and adversarial contributions.
\end{enumerate*}
This decomposition allows us to isolate the economic impact of withholding behavior on each stakeholder group under the PPS scheme.

\textbf{Coinbase rewards:} During the first epoch (spanning $\delta^{BWH}_{p} \tau_0$), $D_0$ canonical blocks are generated. Note that only ($1 - \alpha p$) fraction of the total network hashpower contributes to mining these blocks, with contributions distributed as follows:
The adversary (solo) mines at rate $\alpha(1 - p)$.
The victim pool operator mines at rate $\beta$.
The remaining honest miners mine at rate ($1 - \alpha - \beta$). Thus, the total coinbase reward received by each group during this epoch is
\begin{align}
    C_{1,A_S}&=\frac{\alpha(1-p)}{1-\alpha p}=\alpha(1-p)\delta^{BWH}_{p},\\
    C_{1,P_o}&=\frac{\beta}{1-\alpha p}=\beta\delta^{BWH}_{p},\\
    C_{1,R}&=\frac{1-\alpha-\beta}{1-\alpha p}=(1-\alpha-\beta)\delta^{BWH}_{p}.
\end{align}

\textbf{PPS submission rewards:} Since block difficulty remains unchanged until the end of the first epoch (which lasts $\delta^{BWH}_{p} \tau_0$), the victim pool operator distributes PPS rewards to participants based on their submitted shares. The adversary, contributing fraction $p$ of its total hashpower to the victim pool during this interval, receives
\begin{align}
    P_{1,A_P} = \delta^{BWH}_{p} \alpha p (1 - 2^{-Z})
\end{align}
PPS rewards over the epoch. Note that $2^{-Z}$ fraction of the shares generated by the adversary are withheld fPoW, resulting in a small net loss for the adversary relative to an honest mining scenario. Similarly, the victim pool operator distributes
\begin{align}
    P_{1,H_P} = \delta^{BWH}_{p} \beta
\end{align}
PPS rewards to its honest members during this epoch.

Hence, the total reward received by the adversary in the first epoch is $C_{1,A_S} + P_{1,A_P}$ which is $\delta^{BWH}_{p}\alpha p2^{-Z}$ less than its expected revenue had it mined honestly during $[0, \delta^{BWH}_{p} \tau_0]$. Therefore, the net revenue change for the adversary is $-\delta^{BWH}_{p}\alpha p2^{-Z}$. Similarly, compared to the scenario in which the adversary mines honestly throughout the first epoch, the total reward received by honest members of the victim pool (as well as those honest miners outside the pool) remains unchanged. Their respective rewards are unaffected by the adversary’s withholding behavior during this interval.

Note that, the victim pool operator receives only coinbase rewards equal to $C_{1,P_o} = P_{1,H_P}$, yet it pays out a total of $P_{1,A_P} + P_{1,H_P}$ in PPS payments for submitted shares. In other words, during the first epoch, the victim pool operator bears an additional cost, specifically, it compensates the adversary directly for every pPoW share submitted, effectively paying out-of-pocket to offset the withholding behavior.
\end{Proof}

\begin{corollary}\label{cor::after_DA}
Under the PPS scheme, if the adversary initiates a classical BWH against a victim pool with $p$ fraction of its power from $t=0$ onwards, after the first epoch, the revenue change of each mining entity is given by,
\begin{align}
    \Delta_A^{BWH,H}&(t_1+x)=\frac{\alpha x}{\tau_0}\big(\delta^{BWH}_{p}(1-p2^{-Z})-1\big)\nonumber\\&-\alpha\delta^{BWH}_{p}p2^{-Z}\approx\frac{\alpha x}{\tau_0}(\delta^{BWH}_{p}-1),\\
    \Delta_{H_P}^{BWH,H}&(t_1+x)=\frac{\beta x}{\tau_0}(\delta^{BWH}_{p}-1),\\
    \Delta_{H_R}^{BWH,H}&(t_1+x)=\frac{(1-\alpha-\beta)x}{\tau_0}(\delta^{BWH}_{p}-1),
\end{align}
whereas the revenue change of the pool operator is, 
\begin{align}
    \Delta_{Po}^{BWH,H}(t_1+x)=-\Big(1+\frac{x}{\tau_0}\Big)\alpha \delta^{BWH}_{p}p(1-2^{-Z}).
\end{align}
\end{corollary}

\begin{Proof}
After the first epoch, difficulty is adjusted downward (by $\delta^{BWH}_{p}$), and since the adversary continues its classical BWH attack unchanged, each subsequent epoch lasts $\tau_0$ time units. Without the attack, the reward rate for each entity over each $\tau_0$ interval would be: $\alpha$ for the adversary, $\beta$ for honest miners in the victim pool, $(1 - \alpha - \beta)$ for all other honest miners. With the attack active during each subsequent epoch:
\begin{itemize}
    \item The coinbase rewards received by each entity are $C_{1,A_S}$ (solo adversary), $C_{1,P_o}$ (victim pool operator), and $C_{1,R}$ (other honest miners).
    \item Due to the difficulty reduction by $\delta^{BWH}_{p}$, the victim pool operator pays $P_{1,A_P}$ in PPS rewards to the adversary over each $\tau_0$ interval, whereas prior to this adjustment, the adversary received $P_{1,A_P}$ over a duration of $\delta^{BWH}_{p} \tau_0$.
\end{itemize}
Hence, the revenue change for the adversary at time $t_1 + x$ (i.e., after $x/\tau_0$ subsequent epochs) is,
\begin{align}
    \Delta_A^{BWH,H}(t_1+x)=\big(C_{1,A_S}+P_{1,A_P}-\alpha\big)\frac{x}{\tau_0}+\Delta_A^{BWH,H}(t_1).
\end{align}
The revenue changes for the other entities can be derived similarly. For the victim pool operator, note that over each $\tau_0$ interval after the first epoch, it incurs a net out-of-pocket loss of $P_{1,A_P}$, because these are shares submitted by the adversary that never result in coinbase rewards.
\end{Proof}

\begin{theorem}\label{thm:AoA}
If $\alpha > 2^{-Z}$, \textbf{All-out Attack}, i.e., picking $p^* = 1$ in BWH, maximizes the adversarial revenue change after the first DA under the PPS scheme. Otherwise, the adversary should mine honestly.
\end{theorem}

\begin{Proof}
The adversary’s revenue change in the first epoch is  
\begin{align}
-\frac{\alpha p}{1-\alpha p}\times 2^{-Z},
\end{align}
which represents a small loss since $2^{-Z} \ll 1$. After the first DA, the rate of change of adversarial revenue (its slope) is   
\begin{align}
\frac{\alpha}{\tau_0} \left( \delta^{BWH}_{p}(1 - p 2^{-Z}) - 1 \right),
\end{align} 
which is positive if and only if  
\begin{align}
\delta^{BWH}_{p} > \frac{1}{1 - p 2^{-Z}},
\end{align}
equivalent to $\alpha > 2^{-Z}$. Moreover, under the condition $\alpha>2^{-Z}$, this slope is increasing in $p$. Therefore, to maximize revenue change after the first adjustment, the adversary should set $p = 1$, i.e., adopt the All-out Attack. If $\alpha \leq 2^{-Z}$, mining honestly is optimal.
\end{Proof}

\begin{remark}
Note that $\alpha > 2^{-Z}$ holds in practical systems with block interarrival time $T_b \gg T_{\mathrm{share}}$, where $T_{\mathrm{share}}$ denotes the target share-submission interval. Indeed, if
\begin{align}
2^{-Z}=\alpha,
\end{align}
then a miner with hashpower fraction $\alpha$ can submit shares only at the network block-generation rate (e.g., one share every $10$ minutes in Bitcoin). Since pools typically require shares at least every few seconds, they choose a lower share difficulty,
\begin{align}
2^{-Z}\sim
\alpha\frac{T_{share}}{T_b}.
\end{align}
\end{remark}

\begin{remark}\label{rem::fees}
    If the pool manager charges fees such that a miner receives $r R_s$ rewards for each share that is worth $R_s$ with $0\leq r\leq 1$, under a fixed relative share difficulty, the post-DA adversarial revenue rate becomes $\alpha\frac{1-cp}{1-\alpha p}$ with $c=1-r(1-2^{-Z})$, which is monotone in $p$. Hence, regardless of whether the honest baseline is solo ($\alpha$) or pool mining ($r\alpha$), All-out Attack is optimal if and only if $\alpha>c$, i.e., roughly when $\alpha$ exceeds the pool fee $1-r$, and $p^*=0$ otherwise. The resulting relative revenue change per epoch is $\frac{\alpha-c}{1-\alpha}$ with respect to solo mining and $\frac{r(\alpha-2^{-Z})}{1-\alpha}$ with respect to pool mining.
\end{remark}

\begin{remark}
    The analysis is presented for Bitcoin-style epoch-based DAAs for analytical simplicity. The attack is expected to be at least as effective under faster difficulty adjustment algorithms, since profitability arises immediately after the first downward adjustment in mining difficulty. Faster-reacting DAAs reduce the duration of the initial transient during which the attacker incurs negligible withholding losses, thereby increasing the attack's practical profitability.
\end{remark}

\begin{remark}
Assuming $\alpha > 2^{-Z}$ and neglecting the small fractional loss incurred by the adversary (i.e., assuming $2^{-Z} \approx 0$), the relative revenue change per epoch under AoA after the first DA for each mining entity is $\delta^{BWH}_{1} - 1$. Therefore, all mining entities benefit by the same magnitude per unit of hashpower they control. Conversely, the victim pool operator incurs a loss of $\alpha \cdot \delta^{BWH}_{1}$ in each epoch. It is straightforward to show that the adversary can distribute its hashpower across multiple PPS mining pools without affecting the effectiveness of the attack (i.e., $\beta$ may represent the combined power of multiple pools). For analytical simplicity, we consider a single target pool. Such a distribution may be preferred by the adversary to conceal the attack.
\end{remark}

\begin{remark}
In longest-chain protocols without an uncle block mechanism, releasing withheld fPoW after honest miners mine a competing block induces a fork. This action marginally mitigates the pool operator’s losses while simultaneously reducing the rewards received by external honest miners, since only one of the two blocks in the fork is ultimately rewarded. If the blockchain protocol incorporates uncle blocks into its difficulty adjustment mechanism, releasing withheld fPoW to induce a fork also reduces the magnitude of the difficulty decrease after the first epoch, in turn, diminishes adversarial profits.
\end{remark}

\subsection{Variable Share Difficulty (Vardiff)}
\begin{theorem}\label{thm:AoA_vardiff}
Under the PPS scheme with vardiff, let
\begin{align}
q \triangleq \frac{T_{\rm share}}{T_b}.
\end{align}
Then, the adversarial revenue change under BWH is maximized by
\begin{align}
p^*=\min\left\{1,\frac{1}{\alpha+2q}\right\}.
\end{align}
In particular, if $\alpha<1-2q$, the optimal strategy is the All-out Attack.
\end{theorem}

The proof of Theorem~\ref{thm:AoA_vardiff}, together with the derivation of the optimal revenue change slope, is provided in Appendix~\ref{sec::vardiff}.

\begin{corollary}
Under the PPS scheme with vardiff, by using arbitrarily many sybil identities, the effective aggregate share submission interval, and hence the effective parameter $q$, can be made arbitrarily small. Consequently, the All-out Attack is optimal, i.e., $p^*=1$.
\end{corollary}

\begin{remark}
Under the standing assumption $\alpha+\beta\le\frac12$, we have $\alpha<\frac12$. Hence, if $q<\frac14$, All-out Attack is the optimal strategy. Since practical blockchain systems satisfy $T_b\gg T_{\rm share}$, we typically have $q\ll1$, and hence the All-out Attack is optimal in practice even without sybil identities. For example, in Bitcoin, if the pool targets one share per second,
\begin{align}
q=\frac{1}{600}\approx1.67\times10^{-3}\ll\frac14,
\end{align}
which is more than two orders of magnitude smaller than the threshold. Even with a target share interval of one minute,
\begin{align}
q=\frac{60}{600}=0.1<\frac14,
\end{align}
the All-out Attack remains optimal.
\end{remark}

\begin{figure}[t]
\captionsetup[subfigure]{aboveskip=0pt,belowskip=9pt}
    \centering
\begin{subfigure}[t]{0.82\columnwidth}
\centering
\includegraphics[width=\textwidth]{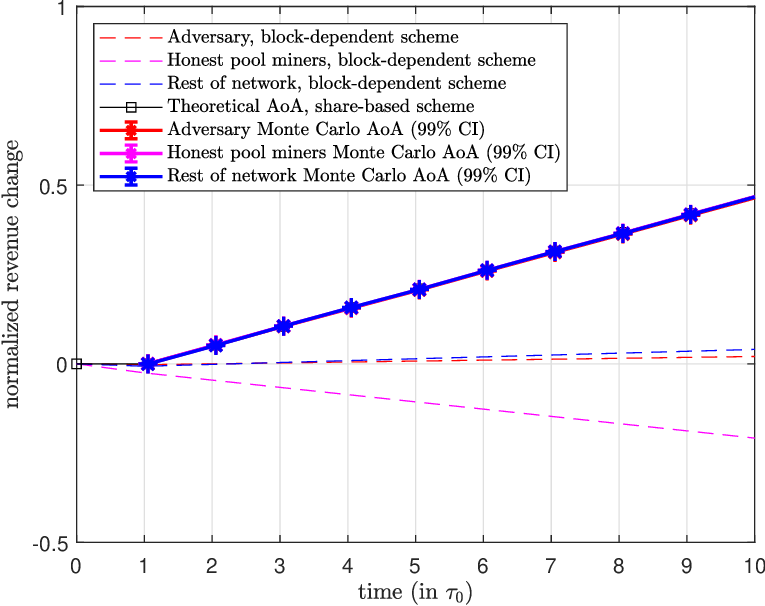}
    \caption{$(\alpha,\beta)=(0.05,0.2)$}
    \label{fig::rev_change_005}
\end{subfigure}
~
\begin{subfigure}[t]{0.82\columnwidth}
    \centering
    \includegraphics[width=\textwidth]{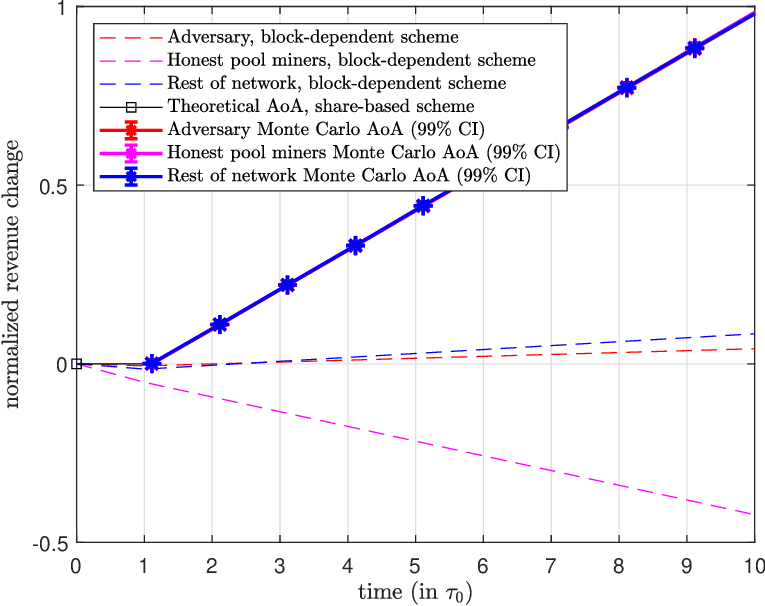}
    \caption{$(\alpha,\beta)=(0.1,0.2)$}
    \label{fig::rev_change_01}
\end{subfigure}
    \caption{Revenue change per unit hashpower in block-dependent and share-based schemes under the BWH attack.}
	\label{fig::rev_change}
\end{figure}

\section{Numerical Results}
We plot the relative revenue change under BWH for both block-dependent and share-based schemes with parameter pairs $(\alpha,\beta) = (0.05, 0.2)$ and $(\alpha,\beta) = (0.1, 0.2)$ in \figref{fig::rev_change}. For block-dependent schemes, we consider the classical BWH variant, where $p$ is chosen to maximize the adversarial revenue ratio as defined in \eqref{eq::adv_rho}. Note that advanced variants of BWH for block-dependent schemes, such as those incorporating fork resolution probabilities, may yield slightly greater adversarial profit than the baseline case presented in \figref{fig::rev_change}. However, even under the state-of-the-art T-PAW attack against the PROP scheme with $\gamma=1$, the \emph{relative extra revenue} (RER) remains strictly below that of the All-out Attack under PPS for every adversarial hashpower fraction $\alpha$ (see~\cite[Fig.~4]{doger2026temporarypoweradjustingwithholding}), where the latter achieves a RER of $\frac{\alpha}{1-\alpha}$.

In \figref{fig::rev_change}, we also present results from Monte Carlo simulations, in which the attack is executed in a sample environment with a vardiff mechanism using $T_{\rm share}=1\,\mathrm{s}$, $T_b=600\,\mathrm{s}$, and $D_0=2016$ blocks. Markers represent sample means computed from $N=1000$ independent simulation runs, while error bars denote $99\%$ confidence intervals for the simulated means. Block arrivals are generated according to a Poisson process, withheld blocks are discarded, and the Bitcoin mining difficulty is updated after each epoch. Share submissions are modeled as an independent Poisson process with target interval $T_{\rm share}$, consistent with the vardiff assumption in Appendix~\ref{sec::vardiff} (see~\cite{mustafadgr_AoA_PPS_vardiff} for implementation details). The simulation results closely match and validate our theoretical analysis.

As shown, share-based schemes exhibit greater vulnerability under BWH and provide a larger profit margin for the adversary. In these schemes, since the relative revenue change per epoch after the first adjustment for each mining entity equals $\delta^{BWH}_{1} - 1 = \frac{\alpha}{1-\alpha}$, the profit per unit hashpower increases monotonically with $\alpha$. 

\section{Conclusion}
We showed that PPS mining pools are not incentive compatible under BWH once difficulty adjustment is taken into account. The optimal attack is the All-out Attack, which raises the adversarial revenue per epoch from $\alpha$ to $\frac{\alpha}{1-\alpha}$ after the first adjustment, while the victim pool operator bears all losses. Remarkably, although the mechanisms differ, this coincides with the upper bound of optimal selfish mining with adversarial network influence $\gamma=1$~\cite{optimal-selfish}, where every adversarial block effectively invalidates a block mined by another entity. Under AoA, the victim pool discovers blocks at only a fraction $\frac{\beta}{\alpha+\beta}$ of its expected rate, a deviation that the pool operator may statistically detect over an epoch; distributing the attack across multiple pools dilutes this deviation per pool, while multiple identities obscure its source. Since the largest Bitcoin pools operate under PPS/FPPS, these results call for pool-side detection mechanisms and payout designs that account for difficulty adjustment, which we leave to future work.

\appendices
\section{Practical Mining Pools: Variable Share Difficulty}\label{sec::vardiff}
If the mining pool employs vardiff to maintain an approximately constant share submission rate, the share difficulty $2^{Z_p}$ assigned to the adversary remains fixed because the adversary's fractional hashrate directed toward the victim pool is assumed to remain constant at $p$, even as the fPoW difficulty changes across epochs. Furthermore, suppose that in the first epoch, the fPoW difficulty is $2^{Z_f}$ with $Z = Z_f - Z_p$. Note that, under vardiff,
\begin{align}
q \triangleq \frac{T_{\rm share}}{T_b},
\end{align}
which is fixed by the pool configuration. Since the pool assigns a share difficulty such that the infiltrating adversary (controlling $\alpha p$ fraction of the total network hashpower) submits one share every $T_{\rm share}$ seconds on average, we have
\begin{align}
2^{-Z}=\alpha p q.
\end{align}
Let $R_{s,j}$ denote the PPS reward per share during epoch $j$. During the first epoch, the reward per share is $R_{s,1}=R_b2^{-Z}$. Since the first difficulty adjustment occurs after the first epoch, the results of Lemma~\ref{lemma::first_DA} also hold in this setting. After the first DA, the share difficulty for adversarial pool mining remains $2^{Z_p}$, since the target share interval is fixed to $T_{\mathrm{share}}$ and the adversary continues to allocate a fraction $p$ of its hashpower toward the victim pool. Meanwhile, due to the increased duration of the first epoch, the fPoW difficulty decreases to
\begin{align}
2^{Z_{f,2}}=\frac{2^{Z_f}}{\delta^{BWH}_{p}}
\end{align}
for the second epoch and subsequent epochs. Consequently, after the first epoch, the reward per share increases to
\begin{align}
R_{s,j}=\delta^{BWH}_{p}R_{s,1}, \quad j\geq2.
\end{align}
This leads to the following result.

\begin{corollary}
Under the PPS scheme (with vardiff), if the adversary initiates a classical BWH against a victim pool with $p$ fraction of its power from $t=0$ onwards, after the first epoch, the revenue change of each mining entity is given by,
\begin{align}
    \Delta_A^{BWH,H}(t_1+x)=&\frac{\alpha x}{\tau_0}\big(\delta^{BWH}_{p}(1-p\delta^{BWH}_{p}2^{-Z})-1\big)\nonumber\\
    &-\alpha\delta^{BWH}_{p}p2^{-Z} \nonumber\\
    \approx&\frac{\alpha x}{\tau_0}(\delta^{BWH}_{p}-1),\\
    \Delta_{H_P}^{BWH,H}(t_1+x)=&\frac{\beta x}{\tau_0}(\delta^{BWH}_{p}-1),\\
    \Delta_{H_R}^{BWH,H}(t_1+x)=&\frac{(1-\alpha-\beta)x}{\tau_0}(\delta^{BWH}_{p}-1),
\end{align}
whereas the revenue change of the pool operator is, 
\begin{align}
    \Delta_{Po}^{BWH,H}(t_1+x)=&-\delta^{BWH}_{p} \alpha p (1-2^{-Z})\nonumber\\&-\Big(\frac{\alpha x}{\tau_0}\Big) \delta^{BWH}_{p} p (1 - \delta^{BWH}_{p}2^{-Z}).
\end{align}
\end{corollary}
\begin{Proof}
The only difference compared with Corollary~\ref{cor::after_DA} is the term $(1-p\delta^{BWH}_{p}2^{-Z})$, which replaces $(1-p2^{-Z})$ in both the adversarial revenue change and the pool operator's revenue change. This difference arises because, under vardiff, the pPoW difficulty assigned to the adversary remains fixed, while the fPoW difficulty decreases after the first difficulty adjustment. As a result, after the first difficulty adjustment, the probability that a share is also a valid fPoW increases from $2^{-Z}$ to $\delta^{\rm BWH}_{p}2^{-Z}$, because the pPoW difficulty remains fixed while the fPoW difficulty decreases. Consequently, the reward per share increases to $R_{s,j}=\delta^{BWH}_{p}R_{s,1}$, $j\geq2$.

Therefore, in each subsequent epoch, the adversary, which contributes a fraction $p$ of its total hashpower to the victim pool, receives
\begin{align}
    P_{j,A_P}
    =\delta^{BWH}_{p}\alpha p
    (1-\delta^{BWH}_{p}2^{-Z}),\quad j\geq2
\end{align}
in PPS rewards over the epoch. 
\end{Proof}

Finally, to prove Theorem~\ref{thm:AoA_vardiff}, note that, after the first DA, the rate of change of adversarial revenue (its slope) is   
\begin{align}
\frac{\alpha}{\tau_0}\big(\delta^{BWH}_{p}(1-p\delta^{BWH}_{p}2^{-Z})-1\big).
\end{align} 
Our goal is to maximize this slope, i.e., the objective can be written as
\begin{align}
f(p)
&=\delta^{\mathrm{BWH}}_{p}
\left(1-p\delta^{\mathrm{BWH}}_{p}2^{-Z}\right)-1 \\
&=\frac{1}{1-\alpha p}
\left(1-\frac{\alpha q p^2}{1-\alpha p}\right)-1 \\
&=\frac{\alpha p\left(1-(\alpha+q)p\right)}
{\left(1-\alpha p\right)^2}.
\end{align}
Differentiating with respect to $p$ yields
\begin{align}
f'(p)
=
\frac{-\alpha\left((\alpha+2q)p-1\right)}
{\left(1-\alpha p\right)^3}.
\end{align}
Since $0\le \alpha p<1$, we have $(1-\alpha p)^3>0$. Hence,
\begin{align}
f'(p)
\begin{cases}
>0, & p<\dfrac{1}{\alpha+2q},\\[1ex]
<0, & p>\dfrac{1}{\alpha+2q},
\end{cases}
\end{align}
which implies that $f(p)$ is increasing for $p<(\alpha+2q)^{-1}$ and decreasing for $p>(\alpha+2q)^{-1}$. Therefore,
\begin{align}
p^*=
\begin{cases}
1, & \alpha+2q<1,\\[2ex]
\dfrac{1}{\alpha+2q}, & \alpha+2q\ge1.
\end{cases}
\end{align}
The corresponding maximum value is
\begin{align}
f(p^*)=
\begin{cases}
\displaystyle
\frac{\alpha(1-\alpha-q)}
{(1-\alpha)^2},
& \alpha+2q<1,\\[3ex]
\displaystyle
\frac{\alpha}{4q},
& \alpha+2q\ge1.
\end{cases}
\end{align}

Next, we note that $q$ can be made arbitrarily small by employing multiple sybil identities when joining the pool. Specifically, suppose the adversary splits its infiltration hashpower $\alpha p$ equally among $m$ sybil identities. Since the pool independently targets a share submission interval of $T_{\rm share}$ for each identity, the aggregate share submission interval becomes $\frac{T_{\rm share}}{m}$, and the effective parameter is reduced to $q_{\rm eff}=\frac{q}{m}$. As $m\to\infty$, we have $q_{\rm eff}\to0$. Since $\alpha<\frac12$ by assumption, there always exists a sufficiently large $m$ such that $\alpha+2q_{\rm eff}<1$. Therefore, the All-out Attack is optimal.

\bibliographystyle{IEEEtran}
\bibliography{blockchain}
\end{document}